\documentclass[12pt,reqno]{amsart}
\usepackage{graphicx}
\usepackage{amsmath}
\usepackage{amssymb}
\usepackage{amscd}
\usepackage{hhline}
\usepackage{mathrsfs}
\usepackage{color}

\newtheorem{theorem}{Theorem}

\newtheorem{prop}[theorem]{Proposition}
\newcommand\bib[1]{\bibitem[#1]{#1}}

\newcommand\1{{\bf 1}}
\renewcommand\a{\alpha}
\renewcommand\b{\beta}

\newcommand\Cc{{\let\mathcal\mathscr\mathcal C}}
\newcommand\Dc{{\let\mathcal\mathscr\mathcal D}}
\newcommand\Hh{{\let\mathcal\mathscr\mathcal H}}
\newcommand\Pp{{\let\mathcal\mathscr\mathcal P}}

\renewcommand\d{\delta}
\newcommand\D{{\mathcal D}}

\newcommand\e{\varepsilon}

\newcommand\g{\mathfrak{g}}

\renewcommand\l{\lambda}

\newcommand\oo{\omega}
\newcommand\op[1]{\mathop{\rm #1}\nolimits}
\newcommand\ot{\otimes}
\newcommand\p{\partial}

\newcommand\R{{\mathbb R}}
\renewcommand\t{\tau}

\newcommand\vp{\varphi}
\newcommand\we{\wedge}

\newcommand\z{\sigma}

\begin{document}

 \title[SDiff(2) and Pleba\'{n}ski II]{SDiff(2) and uniqueness of\\ the Pleba\'{n}ski equation}
 \author{Boris Kruglikov \& Oleg Morozov}
 \date{}
 \address{Institute of Mathematics and Statistics, University of Troms\o, Troms\o\ 90-37, Norway. \quad
E-mails: boris.kruglikov@uit.no; oleg.morozov@uit.no.}

 \vspace{-14.5pt}
 \begin{abstract}
The group of area preserving diffeomorphisms showed importance in the problems of self-dual gravity
and integrability theory. We discuss how representations of this infinite-dimensional Lie group can arise in
mathematical physics from pure local considerations. Then using Lie algebra extensions and cohomology we derive 
the second Pleba\'{n}ski equation and its geometry. We do not use K\"ahler or other additional structures but obtain 
the equation solely from the geometry of area preserving transformations group. We conclude that the Pleba\'{n}ski 
equation is Lie remarkable.
 \end{abstract}

 \maketitle

\section*{Introduction}

Consider a two-dimensional manifold $M$ equipped with an area 2-form.
This structure can be uniformized according to the genus of $M$ and its total area.
In this paper we would like to assume the simplest possible topology
($M$ is $\R^2$ or $\mathbb{S}^2$), concentrating on the geometry of the
group $\op{SDiff}(2)$ of area preserving transformations.

This group can be seen already in the original Pleba\'{n}ski work \cite{P} on the
equations of gravity where he obtained the so-called second heavenly equation
 \begin{equation}\label{PbII}
u_{ty}-u_{xz}+u_{xx}u_{yy}-u_{xy}^2=0,
 \end{equation}
and it played an important role in the subsequent development of the corresponding integrable hierarchies \cite{St,T}.
In this paper we show how this group arises in relation to the second Pleba\'{n}ski equation from a purely local construction; 
for nonlocal structures such as the Lax pair and the recursion operator see \cite{DM,MNS}. The group $\op{SDiff}(2)$ is known 
\cite{MNS,BP} to be related to the classical symmetries of (\ref{PbII}); we shall make this relation two-sided.

For the infinite-dimensional Lie group $\op{SDiff}(2)$ the corresponding Lie algebra $\Dc_0(M)$
consists of divergence free vector fields, which due to trivial topology coincide with
Hamiltonian vector fields. This leads to the classical Lie algebras isomorphism
$(\Dc_0(M),[,])\simeq (C^\infty(M)/\R,\{,\})$, where we use the Lie bracket (commutator) to the left and the Poisson bracket to the right.

The geometry of the Poisson algebra $\Pp=(C^\infty(M),\{,\})$ is central in our paper.
Infinite-dimensional groups are known as an important tool to generate Hamiltonian PDEs
via Euler-Arnold, Gelfand-Dikij and other methods \cite{AK}.
Our strategy is to search for differential invariants among simplest
possible representations of the Lie algebra sheaf of vector fields. Such invariants determine
differential equations that come naturally equipped with a large symmetry algebra.

Since there are no differential invariants in 2 dimensions for $\op{SDiff}(2)$,
we have to extend the Lie algebra or its action.
It turns out that we need to do both, and that the natural extensions yields the action of 4 copies
of $\Pp$ (this will be shown to have the graded structure) on the space of functions in 4 dimensions.
Then the fundamental invariant is the left-hand side of (\ref{PbII}).

To see this we compute the natural differential operators related to the Poisson algebra and
calculate the first cohomology of the Lie algebra $\Pp$ with values in its representations. On this way we
discover the $\op{SDiff}(2)$ analog of the Gelfand-Fuks cocycle, which is central in our method of deriving
the Pleba\'{n}ski equation (\ref{PbII}).

We further explore the symmetry structure of the second Pleba\'{n}ski equation and demonstrate that
it essentially coincides with our extended graded Lie algebra. This makes (\ref{PbII}) the so-called
Lie remarkable equation.

At the end of the paper we shortly discuss the first Pleba\'{n}ski equation, the situation with which
is happened to be quite similar (the integrability properties for both equations are known to be equivalent,
but it is not obvious that the local symmetries and invariants structures are equally reach
since the relation between the equations is non-local).

\medskip

{\textsc{Acknowledgement.}} We thank V.\, Lychagin for useful discussions.

\section{$\op{SDiff}(2)$ and its cohomology}\label{S1}

Let $\Pp$ be the Poisson algebra $(C^\infty(M),\{,\})$, and let $\Dc(M)$ denote the Lie algebra of all
vector fields on $M$. We denote its subalgebra consisting of Hamiltonian fields by $\Dc_0(M)$.

The map $h\mapsto X_h$ is an epimorphism of Lie algebras $\Pp\to\Dc_0(M)$ whose kernel is equal to the
center $\R\subset\Pp$ (on $\mathbb{S}^2$ we can restrict to the space of functions with zero mean).

Since $\op{SDiff}(2)$ acts transitively on $M$, and has open dense orbits in the space of functions on $M$,
to find non-trivial (absolute) invariants we have to consider an extension of the tautological
representation $\Dc_0(M)\subset\op{Der}[C^\infty(M)]$ to a space of bigger dimensions.

Let us start with 1-dimensional extension, i.e. we want to find a homomorphism $\rho$ of $\Dc_0(M)$
to $\Dc(M\times\R)$ such that the vector fields in the image are projectible along $\R$ to
our Hamiltonian fields. In other words, if $\pi\colon M\times\R\to M$ is the natural projection,
then $\pi_*\circ\rho=\1$.

These fields have the form $X_h+\psi(h)\p_u$, where $\p_u$ is the unit vector field along the fiber
coordinate $u$ of $\pi$. The homomorphism condition is equivalent to the claim that $\psi$ is a
1-cocycle on $\Pp$ with values in $C^\infty(M)$. The action on functions is given by
$(h,f)\mapsto X_h(f)=\{h,f\}$ and so is the adjoint representation in $\Pp$.
Thus non-trivial extensions are parametrized by the cohomology group $H^1(\Pp,\Pp)$.

Since all our constructions are required to be local, we will restrict to extensions given by
differential operators. This is also guaranteed by our assumption of trivial topology. Thus in what
follows all cocycles are expressed via differential operators.

 \begin{theorem}\label{thm1}
The above group is 1-dimensional: $H^1(\Pp,\Pp)=\R$. In canonical coordinates $(t,z)$ on $M$
such that $\varpi=dt\we dz$ the generator is represented by the 1-cocycle $\z_1(A)=tA_t+zA_z-2A$.
 \end{theorem}

 \begin{proof}
A linear map $s \colon C^\infty(M)\to C^\infty(M)$ is a 1-cocycle if
 $$
\{s(h),f\}+\{h,s(f)\}=s(\{h,f\}).
 $$
Let $h_{(i,j)}=\D_t^i\D_z^j(h)$, where $\D_t$ is the total derivative by $t$, and similar for $\D_x$.
Writing $s(h)=\sum\l_{ij}h_{(i,j)}$ in the above relation, with the functions $\l_{ij}$
depending on $(t,z)\in M$, we get an expression $\Psi(h,f)=0$ where $\Psi$ is a bilinear
bi-differential operator.

The coefficients of $h_{(1,0)}f_{(i,j)}$ and $h_{(0,1)}f_{(i,j)}$ with $i+j\ne1$
give $\l_{ij}=\op{const}$,
and the coefficients of $h_{(i,j)}f_{(k,l)}$ with $i+j>1$, $k+l>1$, give
$\l_{ij}=0$ for $i+j>1$.

Furthermore the coefficient of $h_{(1,0)}f_{(0,1)}$ gives the equation
 $$
\frac{\partial\l_{10}}{\partial t}+\frac{\partial\l_{01}}{\partial z}=-\l_{00}=\mathrm{const}
 $$
with the solution $\l_{01}=q_t-\frac12\l_{00}z$, $\l_{10}=-q_z-\frac12\l_{00}t$,
where $q=q(t,z)$ is an arbitrary function. This gives the general formula for the cocycle
 $$
s(f)=(q_tf_z-q_zf_t)-\frac12\l_{00}(tf_t+zf_z-2f).
 $$
The first expression in parentheses is the trivial 1-cocycle $\{q,f\}$, while the
second expression in parentheses is not cohomologous to zero.
 \end{proof}

This cocycle allows us to extend the Hamiltonian vector fields on $M=\R^2=J^0(\R,\R)$
to the contact vector fields on $\hat M=M\times\R(u)=J^1(\R,\R)$, namely the new fields are
 $$
X_h+(th_t+zh_z-2h)\p_u=h_z(\p_t+z\p_u)-h_t(\p_z-t\p_u)-2h\p_u
 $$
and an easy change of coordinates brings this field to the canonical form
of the contact Hamiltonian field.

Thus the extension is the standard contact extension, and this algebra can be prolonged
to the algebra of Lie (higher contact) fields in $J^k(\R,\R)$. This however still
acts transitively (no differential invariants) and by this reason in the next section
we extend the manifold $M$ to a 4-dimensional space.

If $M=\mathbb{S}^2$, then we must consider instead the circle bundle $\hat M=PT^*M$
over $M$, and the same arguments work.
In the next section we use only the case $M=\R^2$ to illustrate the argument in coordinates
(the construction is covariant and does not depend on the choice of coordinates,
canonical in the sense of Darboux theorem),
and we do not discuss the counterpart for the sphere.

\section{Extension I: Tangent bundle}\label{S2}

In this section we discuss extension of $M$ by 2 dimensions. There are two natural
candidates: the tangent and the cotangent bundles, and they are isomorphic.

The area form $\varpi=dt\we dz$ on $M=\R^2(t,z)$ induces the isomorphism $TM\simeq T^*M$,
$v\mapsto i_v\varpi$. In the canonical coordinates induced on both bundles by the
coordinates on $M$ this correspondence writes as $(x,y)\leftrightarrow(-y,x)$. We are more
interested in the tangent bundle. The canonical Liouville form from $T^*M$ writes on it
as $\z=x\,dz-y\,dt\in\Omega^1(TM)$.

The canonical symplectic form
 $$
\oo=d\z=dx\we dz+dt\we dy\in\Omega^2(TM)
 $$
is related to the pull-back of $\varpi$ (to 4-dimensional $TM$) by the operator field
$2K=\p_x\ot dt+\p_y\ot dz$, $i_K\oo=\varpi$.

Even more important ingredient is the truncated total field
 $$
\nabla=x\p_t+y\p_z.
 $$
If we identify $J^1(\R,M)=\R\times TM\simeq \R^1(\tau)\times\R^4(t,z,x,y)$,
the total derivative is $\D_\t=\p_\t+x\p_t+y\p_z+\dots$ and we quotient $J^1(\R,M)$ by
the first factor (consider $\t$ independent functions on this jet-space $J^1$).
The field $\nabla$ relates the two symplectic forms as
 $$
i_\nabla\varpi=\z\ \ \Rightarrow\ \ L_\nabla\varpi=\oo.
 $$
In addition we have that $2K(\nabla)=x\p_x+y\p_y$ is the Liouville (radial vertical) field on $TM$.

This $\nabla$ however is not a vector field (as a total differential its tail, namely the part containing
$\p_x,\p_y$ is not uniquely defined, see e.g. \cite{KL}), but only a first order differential operator
 $$
\nabla \colon C^\infty(M)\to C^\infty(TM).
 $$

 \begin{prop}
A linear differential operator $\tilde\nabla \colon C^\infty(M)\to C^\infty(TM)$
of order 1 satisfying
 $$
\tilde\nabla\{A,B\}_\varpi=\{\tilde\nabla A,\tilde\nabla B\}_\oo
 $$
has the form $\tilde\nabla=\nabla+\{q,\cdot\}$, $q\in C^\infty(M)$.
 \end{prop}

 \begin{proof}
The general form of first order differential operator is
 $$
\tilde\nabla=a(t,z,x,y)\p_t+b(t,z,x,y)\p_z+c(t,z,x,y).
 $$
In addition, $a$ and $b$ are not simultaneously zero. Thus substituting $B=1$ we obtain $c=0$.

Next it is easy to see that both $a,b\ne0$. Indeed if e.g. $b=0$, then taking $A=z$, $B=t^2$ we
get a contradiction.

Finally consider the coefficients of $A_{tt}B_t$, $A_{tt}B_z$, $A_{zz}B_t$, $A_{zz}B_z$ in
the defining relation. They imply the system $a_x=1$, $a_y=0$, $b_x=0$, $b_y=1$. The coefficient
of $A_tB_z$ gives $a_t+b_z=0$. This gives the following form
 $$
\tilde\nabla=(x-q_z)\p_t+(y+q_t)\p_z
 $$
for some function $q=q(t,z)$. In other words $\tilde\nabla A=\nabla A+\{q,A\}$.
 \end{proof}

Thus we see that though $\nabla$ is natural, the condition that the operator
preserves the Poisson brackets leads us to consideration of a pair of independent Hamiltonians
on $TM$: $\nabla A_0$ and $A_1$ (the later is lifted from $M$ via pull-back),
$A_i\in\Pp$ (we use the fact that the operator $\{q,\cdot\}:C^\infty(M)\to C^\infty(M)$
is epimorphic; in the case of the sphere with the additional condition that the Hamiltonians have zero mean).

In other words we have the graded Poisson algebra $\Hh_1=\Pp_0\oplus\Pp_1$ (the index refers to the grading),
where $\Pp_i\simeq\Pp$ and the bracket is given by $\{A_i,B_j\}=\{A,B\}_{i+j}$ (and we assume that the
grading 2 is void). This admits the graded Lie algebra homomorphism
 $$
V \colon \Hh_1\to\Dc_1,
 $$
where $\Dc_1=\Dc(TM)_0\oplus\Dc(TM)_1$ is the graded Lie algebra consisting of vector fields
in pure gradings, with the bracket being given by the commutator and the same truncation rule as above.

This homomorphism associates to the element $(A,0)\in\Hh_1$ the Hamiltonian vector field $X_{\nabla A}$,
and to the element $(0,A)\in\Hh_1$ the field $X_{A}$
(both with respect to the symplectic structure $\oo$ on $TM$).

In canonical coordinates $(t,z,x,y)$ on $TM$ we can write this homomorphism as
$(A_0,A_1)\mapsto(V_0(A_0),V_1(A_1))$ with
 \begin{gather*}
V_0(A)=A_z\p_t-A_t\p_z+(A_{tz}x+A_{zz}y)\p_x-(A_{tt}x+A_{tz}y)\p_y,\\
V_1(A)=A_z\p_x-A_t\p_y.
 \end{gather*}
Notice that both vector fields are projectible to Hamiltonian vector fields on $(M,\varpi)$,
with the Hamiltonians $A$ and $0$ respectively.

It is natural to ask if the above Lie algebra homomorphism extends to bigger
truncated graded Lie algebras $\Hh_k=\Pp_0\oplus\dots\oplus\Pp_k$ and
$\Dc_k=\Dc(TM)_0\oplus\dots\Dc(TM)_k$ with the same rule that
$\{A_i,B_j\}=\{A,B\}_{i+j}$ if $i+j\le k$ and $=0$ if $i+j>k$ (and similar in the case
of vector fields: $[V_i,W_j]=[V,W]_{i+j}$ or $=0$ if $i+j>k$).

 \begin{prop}\label{Proposition3}
For $k>1$ the graded Lie algebra homomorphism $V \colon \Hh_k\to\Dc_k$ extending the one for $k=1$
vanishes in all gradings $>1$.
 \end{prop}

 \begin{proof}
Indeed, $V_2(\{A_1,B_1\})=[V_1(A_1),V_1(B_1)]=0$ for all $A_1,B_1\in\Pp_1$.
Since any element $C_2\in\Pp_2$ can be written as $\{A_1,B_1\}$ we get $V|_{\Pp_2}=0$.
Similarly we conclude that $V|_{\Pp_i}$ is trivial for all $i>1$.
 \end{proof}

\section{Extension II: Towards functions in 4D}\label{S3}

Our strategy is to calculate the differential invariants, but for this we need to
extend the action from the space $TM$ to the space $J^0=TM\times\R\simeq\R^4(t,z,x,y)\times\R^1(u)$.

A homomorphism $V\colon \mathcal{G}\to\Dc(TM)$ of a Lie algebra $\mathcal{G}$ 
extends to a homomorphism $\hat V \colon \mathcal{G}\to\Dc(J^0)$,
$\hat V(\alpha)=V(\alpha)+\Psi(\alpha)\p_u$, iff $\Psi$ is a 1-cocycle on $\mathcal{G}$
with values in the module $C^\infty(TM)$ (the action is via the representation $V$).
As usual the 1-cocycles that differ by 1-coboundary define isomorphic extensions.
Thus we need to calculate the cohomology group $H^1(\mathcal{G},C^\infty(TM))$.

In this section we focuss on the simplest case, when $\mathcal{G}$ is $\Hh_0=\Pp$.

 \begin{theorem}\label{thm2}
We have $H^1(\Hh_0,C^\infty(TM))=\R^2$. In the canonical coordinates $(t,x,y,z)$ on $TM$
the cocycles $\z_1(A)=tA_t+zA_z-2A$ and $\z_{\text{GF}}(A) = \nabla^3(A)$  can be taken as a
base for the cohomology group.
 \end{theorem}

 \begin{proof}
Let $\Psi(A)=\sum\l_{ij}A_{(i,j)}$ for some functions $\l_{ij}\in C^\infty(TM)$.
Then the coefficients of $A_{(i,j)}=\D_t^i\D_z^j(A)$, $i+j>2$, in the cocycle identity
 \begin{equation}\label{cocycle_eqn_1}
V_0(A)(\Psi(B))-V_0(B)(\Psi(A))  = \Psi(\{A,B\})
 \end{equation}
yield $\l_{ij} \equiv 0$ for $i+j>3$. The remaining coefficients of $A_{(i,j)},B_{(i,j)}$ with $i+j\le2$
in (\ref{cocycle_eqn_1}) provide an over-determined system of PDEs for the functions $\l_{ij}$, $i+j\le3$.
Taking into account that the solution  $\Psi(A)=\sum \limits_{0\le i+j\le 3}\l_{ij}\D_t^i\D_z^j(A)$
of this system is defined up to adding a coboundary $V_0(A)\, \lrcorner \,  dF$ with an arbitrary $F \in C^\infty(TM)$, 
we obtain the general solution
 \[
\Psi(A) = c_1\,(tA_t+zA_z-2A)+c_2\,\nabla^3(A) + c_3\,(y^{-2}A_{tt}+x^{-2}A_{zz}).
 \]
The requirement $\l_{ij}\in C^\infty(TM)$ implies $c_3=0$.
 \end{proof}

 \bigskip

\noindent
{\bf Remark on Gelfand-Fuks cohomology.}
Calculations of cohomology of infinite-dimensional algebras have started in 1968 with
the Lie algebra $\g=\Dc(\mathbb{S}^1)$ \cite{GF}. In particular, this work introduced
the celebrated Gelfand-Fuks cocycle $c_{\text{GF}}$ as the generator of $H^2(\g)$.

Let us notice that the natural morphism $\delta \colon C^2(\g)\to C^1(\g,\g')$, where the
regular dual $\g'=\mathcal{F}_2=\{f(\vp)d\vp^2\}$ is the space of quadratic differentials,
maps cocycles to cocycles \cite{F}. It induces an isomorphism in cohomology, and $H^1(\g,\g')=\R$
has generator $\d c_{\text{GF}}$ represented by $f(\vp)\p_\vp\mapsto f'''(\vp)d\vp^2$.
On the level of Lie groups $H^1(\op{Diff}(\mathbb{S}^1),\mathcal{F}_2)$ is 1-dimensional and generated by
the Schwarzian derivative \cite{OT}.

The higher-dimensional versions of Schwarzian derivatives exist, and they are cocycles on
$\op{Diff}(M)$ with values in (2,1)-tensor fields. The Lie algebra version in dimension 2
when restricted to the algebra $\Dc_0(M)\subset \Dc(M)$ takes values in $\Gamma(S^3T^*M)$
(notice that $\dim S^3T^*_aM=4$ for $\dim M=2$) and is given by the formula \cite{OT}:
 $$
\hat{c}_{\text{GF}}(X_F)=d^3F.
 $$
This is clearly the 2-dimensional analog of the Gelfand-Fuks cocycle
(we think about 1-cocycle given by the morphism $\d$).

In our case the cocycle $\z_{\text{GF}}=\nabla^3$ takes values in the space of functions on another
4-dimensional space $TM$ (our version gives a lower-dimensional representation of elements of
the Lie algebra by vector fields). Thus it can be considered as the generalized
Gelfand-Fuks 1-cocycle in the case of Lie algebra $\Dc_0(M)$.

Our construction has some similarity with the one in \cite{OT},
which explores the double $\g\oplus\g'$ (followed by passing to the current algebra to
increase the dimension of the configuration space), but in our case $\Hh_0\oplus\Hh_1$
the second summand is adjoint (not co-adjoint) module and (what is more important) all 
extensions do satisfy the Lie pseudogroup property: they come with natural representation by 
(the sheaf of) vector fields and are given by determining differential equations.

\section{Extension III: Formal series and natural truncation}\label{S4}

Now we consider the case, when $\mathcal{G}$ is $\Hh_1=\Pp_0\oplus\Pp_1=\{(A_0,A_1)\}$.
The same computations as in theorem \ref{thm2} give

 \begin{theorem}\label{thm3}
$H^1(\Hh_1,C^\infty(TM))=\R^2$, and the following two cocycles form its basis:
$\z_1(A_0)=tA_{0,t}+zA_{0,z}-2A_{0}$ and $\z_2(A_1) = A_1$.
 \end{theorem}

This result seems to be rather disappointing, since the most interesting cocycle $\nabla^3(A_0)$
disappears after passing from $\Hh_0=\Pp$ to $\Hh_1=\Pp\ot\R[[\e]]/\{\e^2=0\}$.
The reason is the cut tails in the series.

To overcome the problem, we consider the Lie algebra of formal series
$\Hh_\infty=\Pp_0\oplus\Pp_1\oplus\dots\oplus\Pp_k\oplus\dots =\Pp\ot\R[[\e]]$.
We want to find an extension of the homomorphism $V \colon \Hh_\infty\rightarrow \Dc(TM)$ to
 $$
\hat{V} \colon \Hh_\infty\to\Dc(J^0)
 $$
(recall that $J^0=TM\times\R$).
This is given by 1-cocycle $\Psi$ on $\Hh_\infty$ with values in $C^\infty(TM)$
via representation $V$. By Proposition \ref{Proposition3} this latter is equal to
 $$
V \colon A=\sum\limits_{k=0}^{\infty}\e^k A_k \mapsto V_0(A_0)+\,V_1(A_1).
 $$
So the defining relation $V(A_i)\Psi(B_j)-V(B_j)\Psi(A_i)=\Psi(\{A,B\}_{i+j})$
implies that the 1-cocycle $\Psi$ on $\Hh_\infty$ vanishes in grading $>3$.
The same computations as in theorems \ref{thm1} and \ref{thm2} yield

 \begin{theorem}\label{thm3!}
The following 1-cocycles form a basis in the cohomology group $H^1(\Hh_\infty,C^\infty(TM))$:
 \begin{eqnarray*}
\zeta_1(A) &=& \tfrac16\nabla^3(A_0) + \tfrac12\,\nabla^2(A_1)+\nabla(A_2)+A_3, \\
\zeta_2(A) &=& \nabla(A_1)+2\,A_2, \\
\z_1(A) &=& t\,A_{0,t}+z\,A_{0,z} -2\,A_0, \\
\z_2(A) &=& A_1.
 \end{eqnarray*}
 \end{theorem}
The required extension is given by the formula ($c_i =\mathrm{const}$):
 $$
\hat{V} \colon \sum\limits_{k=0}^{\infty}\e^k A_k\mapsto
\hat{V}_0(A_0) + \hat{V}_1(A_1)+ \hat{V}_2(A_2)+ \hat{V}_3(A_3),
 $$
 \begin{eqnarray*}
\hat{V}_0(A_0) &=& V_0(A_0) +
\Bigl(\tfrac16c_1\nabla^3(A_0) + c_3 (t A_{0,t} + z A_{0,z} - 2 A_{0})\Bigr)\,\p_u \\
\hat{V}_1(A_1) &=& V_1(A_1) + \Bigl(\tfrac12\,c_1\nabla^2(A_1) +c_2\nabla(A_1) +c_4 A_1\Bigr)\,\p_u \\
\hat{V}_2(A_2) &=& \bigl(c_1 \nabla(A_2)+2\,c_2 A_2\bigr)\,\p_u \\
\hat{V}_3(A_3) &=& c_1 A_3\,\p_u.
 \end{eqnarray*}

Now we shall classify the family $\frak{h}_{(c_1,c_2,c_3,c_4)} = \hat{V}(\Hh_\infty)$
up to an isomorphism, preserving the filtration
$\oplus_{i\ge k}\Pp_i=\Pp\ot\e^k\R[[\e]]\subset\Hh_\infty$.

 \begin{theorem}\label{thm4}
1) When $c_1 \not =0$, there exists an isomorphism
 $$
\Phi_1 \colon  \frak{h}_{(c_1,c_2,c_3,c_4)} \rightarrow \frak{g}_1 = \frak{h}_{(1,0,0,0)}
 $$
defined as a superposition of the map (the function $\z_1$ below
is the same as in Theorems \ref{thm3} and \ref{thm3!})
 \begin{eqnarray*}
\hat{V}_0(A_0) &\mapsto& \hat{V}_0(A_0)- \frac{c_3}{c_1}\,\hat{V}_3\circ\z_1(A_{0}), \\
\hat{V}_1(A_1) &\mapsto& \hat{V}_1(A_1)-\frac{c_1c_4-2c_2^2}{3c_1}\,\hat{V}_3(A_1)
  -\frac{c_2}{c_1}\,\hat{V}_2(A_1), \\
\hat{V}_2(A_2) &\mapsto& \hat{V}_2(A_2)- \frac{2c_2}{c_1}\,\hat{V}_3(A_2), \\
\hat{V}_3(A_3) &\mapsto& \hat{V}_3(A_3).
 \end{eqnarray*}
and the scaling $u\mapsto -c_1u$, $A_2\mapsto -A_2$, $A_3\mapsto -A_3$.

\medskip

2) If $c_1=0$, $c_2\not=0$, then $\hat{V}_3(A_3)=0$ and the map
 \begin{eqnarray*}
\hat{V}_0(A_0) &\mapsto& \hat{V}_0(A_0) - \frac{c_3}{2 c_2}\,\hat{V}_2\circ\z_1(A_{0}), \\
\hat{V}_1(A_1) &\mapsto& \hat{V}_1(A_1) - \frac{c_4}{2 c_2}\,\hat{V}_2(A_1), \\
\hat{V}_2(A_2) &\mapsto& \hat{V}_2(A_2)
 \end{eqnarray*}
together with the scaling $u\mapsto c_2u$ defines an isomorphism
 $$
\Phi_2 \colon  \frak{h}_{(0,c_2,c_3,c_4)} \rightarrow \frak{g}_2 = \frak{h}_{(0,1,0,0)}.
 $$

\smallskip

3) In the case of $c_1=c_2=0$ we have $\hat{V}_2(A_2)=\hat{V}_3(A_3)=0$. The Lie algebra
$\frak{g}_3=\frak{h}_{(0,0,c_3,c_4)}$ is defined up scaling of $(c_3,c_4)$,
and it coincides with the extension from Theorem \ref{thm3}.
 \end{theorem}

Let us denote the generators $\hat{V}_i$ in case {\it1)\/} by $W_i$,
 \begin{eqnarray*}
W_0(A_0) = V_0(A_0) - \tfrac{1}{6}\,\nabla^3(A_0)\,\p_u, &&
W_2(A_2) = \nabla(A_2)\,\p_u, \\
W_1(A_1) = V_1(A_1) -\tfrac{1}{2}\,\nabla^2(A_1)\,\p_u, &&
W_3(A_3) = A_3\,\p_u.
 \end{eqnarray*}

The Lie algebra $\frak{g}_1= \frak{a}_0\oplus \frak{a}_1\oplus \frak{a}_2\oplus \frak{a}_3$ is
4-graded, $[\frak{a}_i,\frak{a}_j] =\frak{a}_{i+j}$.
Here $\frak{a}_i= \{W_i(A_i) \,|\, A_i\in C^\infty(M)\}$, and $\frak{a}_i=0$ for $i\not\in\{0,1,2,3\}$.
In addition, $\Phi_1$ is a graded Lie algebra homomorphism.

Similarly, the Lie algebra $\frak{g}_2= \frak{a}_0\oplus \frak{a}_1\oplus \frak{a}_2$ is
3-graded, and $\Phi_2$ is a graded Lie algebra homomorphism. 

Finally, the Lie algebra
$\frak{g}_3= \frak{a}_0\oplus \frak{a}_1$ is 2-graded.

\section{Differential invariants of the action}\label{S5}

Let $\frak{G}_1$, $\frak{G}_2$, $\frak{G}_3$ be the Lie pseudo-groups on $J^0=J^0(TM,\R)$
with the Lie algebras $\frak{g}_1$, $\frak{g}_2$ and $\frak{g}_3$, respectively.
By direct computations, using \textsc{Maple}, we find differential invariants of the prolongations
of actions of these pseudo-groups on $J^2(TM,\R)$:

 \begin{theorem}\label{thm5}
The only differential invariants of the action on $J^2(TM,\R)$ are:
 \begin{itemize}
\item[{\it1)}] $\frak{G}_1$: \ $I_1=u_{ty}-u_{xz}+u_{xx}u_{yy}-u_{xy}^2$.
 \medskip
\item[{\it2)}] $\frak{G}_2$: \ $I_2= (u_y-x)^2 u_{xx} - 2\,(u_x+y)(u_y-x)\,u_{xy} +(u_x+y)^2 u_{yy}$ and
  $I_3= u_{xx}u_{yy} - u_{xy}^2$.
 \medskip
\item[{\it3)}] $\frak{G}_3$: \ the above $I_3$ and\/ $I_4 = u_y^2 u_{xx} - 2\,u_xu_y\,u_{xy} +u_x^2 u_{yy}$.
 \end{itemize}
 \end{theorem}

The equation $I_1= 0$ is the second Pleba\~{n}ski equation (\ref{PbII}).
The equations $I_2=0$, $I_3=0$, $I_4=0$ have only two independent variables and are less interesting.
Moreover, they can be linearized by contact transformations.
For the Monge-Amp\`ere equation $I_3=0$ this result is classical. For two other equations it can be
proven by the methods of \cite{M} or \cite{KLR}:

 \begin{theorem}\label{thm6}
Equations
\[
(u_y-x)^2 u_{xx} - 2\,(u_x+y)(u_y-x)\,u_{xy} +(u_x+y)^2 u_{yy} = 0
\]
and
\[
u_y^2 u_{xx} - 2\,u_x u_y u_{xy} +u_x^2 u_{yy} = 0
\]
are contact-equivalent to the equation
\[
\widetilde{u}_{\widetilde{x}\widetilde{x}} = 0.
\]
 \end{theorem}

The only differential invariant of the prolongation of action of
the Lie pseudogroup $\frak{G}_1$ on $J^3$ is the function
 \[
J_1=E_2E_4 - E_1E_3 - u_{xx}E_3^2+2\,u_{xy}E_2E_3 -u_{yy}E_2^2,
 \]
where $E_1=\D_t(I_1)$, $E_2=\D_x(I_1)$, $E_3=\D_y(I_1)$, $E_4=\D_z(I_1)$.
The invariant derivations of the prolongation of $\frak{G}_1$ on jets of order greater than 2 are
 \begin{eqnarray*}
\mathbb{D}_1 &=& E_3 \D_x - E_2 \D_y, \\
\mathbb{D}_2 &=& E_3 \D_t + E_4 \D_x - E_1 \D_y - E_2 \D_z, \\
\mathbb{D}_3 &=& (u_{yy}E_2-u_{xy}E_3-E_4) \,\D_x - (u_{xy}E_2-u_{xx}E_3-E_1)\, \D_y, \\
\mathbb{D}_4 &=& (u_{yy}E_2-u_{xy}E_3-E_4)\, \D_t \\
  &&- (u_{yy}E_1+(u_{xx}u_{yy}-u_{xy}^2)E_3-u_{xy}E_4)\, \D_x \\
  && + (u_{xy}E_1+(u_{xx}u_{yy}-u_{xy}^2)E_2-u_{xx}E_4)\, \D_y \\
  && -(u_{xy} E_2-u_{xx}E_3-E_1)\, \D_z.
 \end{eqnarray*}
We have $\mathbb{D}_1(I_1)=\mathbb{D}_2(I_1)=\mathbb{D}_4(I_1)=0$ and $\mathbb{D}_3(I_1)=-J_1$.
Let $c_{ij}^k$ be the structural functions of the frame $\mathbb{D}_1$, \dots , $\mathbb{D}_4$,
 \[
[\mathbb{D}_i, \mathbb{D}_j]  = c_{ij}^k\,\mathbb{D}_k,
\qquad 1\le i < j \le 4.
 \]
Then $c_{ij}^k$ are rational functions of $J_1$ and the fucntions $K_1$, \dots , $K_{11}$ defined as
$K_1 = c_{12}^1$, $K_2 = c_{12}^2$, $K_3 = c_{12}^4$,
$K_4 = c_{13}^4$, $K_5 = c_{14}^3$, $K_6 = c_{23}^4$,
$K_7 = c_{24}^1$, $K_8 = c_{24}^2$, $K_9 = c_{24}^3$,
$K_{10} = c_{24}^4$ and $K_{11} = c_{34}^1$. These together with $I_1,J_1$ and $\mathbb{D}_i(J_1)$
form a basis of differential invariants on 4-jets.

As our Maple computations indicate, the whole algebra of scalar (absolute) differential invariants
of $\frak{G}_1$ is generated by the fundamental invariant $I_1$
and the invariant derivations $\mathbb{D}_i$, i.e. the iterated invariant derivatives of
$I_1$ and functions of them (in particular all $K_j$ are obtained so) yield all the invariants.

However the generators $\mathbb{D}_1$, \dots , $\mathbb{D}_4$ vanish on equation (\ref{PbII}), and
thus $\{I_1=0\}$ is a singular manifold of the action of $\frak{G}_1$ on $J^\infty(TM,\R)$.

\section{Symmetries of Pleba\'{n}ski II}\label{S6}

As we have shown, the second Pleba\'{n}ski equation arises naturally from the Lie algebra
$\frak{g}_1$ (which in turn is a natural extension of $\op{SDiff}(2)$). On the other hand,
this algebra appears to be an infinite-dimensional part of the algebra of contact symmetries
of (\ref{PbII}) ---  the following statement is obtained by a direct computation, cf. \cite{MNS}.

 \begin{theorem}\label{thm7}
The Lie algebra of contact symmetries of equation (\ref{PbII}) is the graded Lie algebra
$\tilde{\frak{g}}_1= \tilde{\frak{a}}_0\oplus \tilde{\frak{a}}_1\oplus \tilde{\frak{a}}_2\oplus \tilde{\frak{a}}_3$, $[\tilde{\frak{a}}_i,\tilde{\frak{a}}_j]\subset \tilde{\frak{a}}_{i+j}$ with
$\tilde{\frak{a}}_0 =\frak{a}_0 \oplus \R\cdot W_0'\oplus\R\cdot W_0''$,
$\tilde{\frak{a}}_1 =\frak{a}_1 \oplus \R\cdot W_1'$,
$\tilde{\frak{a}}_2 =\frak{a}_2$, $\tilde{\frak{a}}_3 =\frak{a}_3$
{\rm(}$\tilde{\frak{a}}_i = 0$ for $i\not\in\{0,1,2,3\}${\rm)}. Here
 \begin{eqnarray*}
W_0' &=& t\,\p_t+x\,\p_x+y\,\p_y+z\,\p_z+2\,u\,\p_u,\\
W_0''&=& -x\,\p_x-y\,\p_y-3\,u\,\p_u, \qquad\qquad\quad
W_1'  =  t\,\p_x+z\,\p_y.
 \end{eqnarray*}
The structure equations of $\tilde{\frak{g}}_1$ are the following
(the functions $\z_1,\z_2$ are the same as in Theorem \ref{thm3}):
 \begin{gather*}
[W_i(A_i),W_j(A_j)]=W_{i+j}(\{A_i,A_j\}),\\
[W_0',W_i(A_i)]=W_i(\z_1(A_i)),\quad [W_0'',W_i(A_i)]=i\,W_i(A_i),\\
[W_1',W_i(A_i)]=W_{i+1}((\z_1+i\,\z_2)(A_i)),\\
[W_0',W_0'']=0,\ [W_0',W_1']=0,\ [W_0'',W_1']=W_1'.
 \end{gather*}
 \end{theorem}

Thus the descending central and derived series of $\tilde{\frak{g}}_1$ are:
$[\tilde{\frak{g}}_1,\tilde{\frak{g}}_1] = [\tilde{\frak{g}}_1,\frak{g}_1\oplus\R\cdot W_1'] =
\frak{g}_1\oplus\R\cdot W_1'$,
$[\frak{g}_1\oplus\R\cdot W_1',\frak{g}_1\oplus\R\cdot W_1'] = [\frak{g}_1,\frak{g}_1] = \frak{g}_1$.

Moreover the algebra $\tilde{\frak{g}}_1$ is restored from $\frak{g}_1\simeq\Hh_3$ in two steps.
At first we apply the 2-dimensional right extension by the derivations
$\e^kA_k\mapsto\e^k\z_1(A_k)$, $\e^kA_k\mapsto\e^{k+1}(\z_1+k\,\z_2)(A_k)$ of degrees 0 and 1.
The corresponding cohomology classes in $H^1(\Hh_3,\Hh_3)$ are closely related to the
fundamental cohomology classes from Theorems \ref{thm3} and \ref{thm3!}.

Then we do 1-dimensional extension by the grading element $W_0''$.

It is important to stress that if we stop on the first step, we obtain the full symmetry algebra
$\g\oplus\R\cdot W_0'\oplus \R\cdot W_1'$ of the function $I_1$. 
The remaining field does not preserve $I_1$ -- it is a relative differential invariant for $W_0''$.
Indeed its second prolongation satisfies: $\op{pr}_2(W_0'')(I_1)=-2\,I_1$.

\medskip

\noindent
 {\bf Remark on Lie remarkable property.}
Finite and infinite-dimen\-sional Lie algebras of classical symmetries
are important in integration and establishing exact solutions of differential equations.

On the other hand for any Lie pseudogroup of symmetries we can calculate its prolongation
to the space of $k$-jets and consider non-trivial orbits for the smallest $k$,
which can be considered geometrically as differential equations.
If these two processes are inverse of each other, the equation is called
Lie remarkable (the original paper \cite{MOV} deals with point symmetries, i.e.
fields on the space of 0-jets $\Dc(J^0)$, but it extends to
the contact fields on the space of 1-jets $\mathfrak{cont}(J^1)$).

In the particular case of scalar determined equation (one independent variable and one PDE)
we calculate the symmetry group and (if it is non-trivial) look for the lowest order
differential invariant $I\in C^\infty(J^k)$. If it is unique (up to the gauge $I\mapsto F(I)$),
and $I=0$ coincides with our PDE, the latter has the above property. Thus the Pleba\'{n}ski equation
is Lie remarkable (in general the equation $I=c$ can depend on the value of $c$, and there can be
even regular and non-regular values, but for (\ref{PbII}) this constant can be easily absorbed).

Not all equations are Lie remarkable. For instance, the Boyer-Finley equation $u_{z\bar z}=(e^u)_{tt}$
(another 'heavenly' equation) has 5 differential invariants of order 2 (3 of pure order 2)
of its symmetry groups action \cite{Sh} (in this case the group $G$ is also infinite-dimensional,
and it consists of conformal transformations of $\R^2$ together with a translation and a scaling).
This makes possible application of the method of group foliation, but it does not
uniquely characterize the equation.

\appendix

\section{Pleba\'{n}ski I equation}\label{SA}

Let us briefly discuss the structure of the contact symmetry algebra of the first Pleba\'{n}ski's
heavenly equation \cite{P}
 \begin{equation}\label{PbI}
u_{tx}u_{yz} - u_{tz}u_{xy} =1,
 \end{equation}
which was studied in \cite{BW,MNS}. It turns out that the infinite part of this
symmetry algebra is also composed of 4 copies of $\op{SDiff}(2)$, but now it is 2-graded,
to be more precise it is a copy of two such algebras.

Thus instead of $TM$ we get $M\times M$, with the symplectic form being the product structure.
Let $M_1=\R^2(t,y)$, $M_2=\R^2(x,z)$ be the two copies of $M$ with $\varpi_1=dt\we dy$, $\varpi_2=dx\we dz$.
These generate the Lie sub-algebra
 $$
\mathfrak{b}_0=\{X_{A_0}+X_{B_0}\,|\,A_0\in C^\infty(M_1),B_0\in C^\infty(M_2)\}\subset\Dc(M\times M)
 $$
consisting of two copies of $\op{SDiff}(2)$. Letting $J^0=J^0(M\times M)=M\times M\times\R$, with
the coordinate $u$ on the last factor, we extend the algebra to include two more copies of $\op{SDiff}(2)$:
 $$
\mathfrak{b}_1=\{(A_1+B_1)\,\p_u\,|\,A_1\in C^\infty(M_1),B_1\in C^\infty(M_2)\}\subset\Dc(J^0).
 $$
One can easily check that $\g=\mathfrak{b}_0\oplus\mathfrak{b}_1$ is a graded Lie algebra.

To indicate the grading we will write the generators of the algebra $\g$ as
$Y_0^\a(A_0)=X_{A_0}$, $Y_0^\b(B_0)=X_{B_0}$, $Y_1^\a(A_1)=A_1\,\p_u$, $Y_1^\b(B_1)=B_1\,\p_u$.

\begin{theorem}\label{thm8}
The Lie algebra $\tilde\g$ of contact symmetries of equation (\ref{PbI}) is equal to
$\g\oplus\R^3\langle Y_0',Y_0'',\tilde{Y}_0\rangle$, where
 $$
Y_0'  = t\,\p_t-y\,\p_y,\quad
Y_0'' = x\,\p_x-z\,\p_z,\quad
\tilde{Y}_0= t\,\p_t+y\,\p_y-x\,\p_x-z\,\p_z.
 $$
It is graded by
$\tilde{\mathfrak{b}}_0=\mathfrak{b}_0\oplus\R^3\langle Y_0',Y_0'',\tilde{Y}_0\rangle$,
$\tilde{\mathfrak{b}}_1=\mathfrak{b}_1$. The structure equations of $\tilde\g$ are:
 \begin{gather*}
[Y_i^\a(A_i),Y_j^\a(\bar{A}_j)]=Y_{i+j}^\a(\{A_i,\bar{A}_j\}),\quad
  [Y_i^\a(A_i),Y_j^\b(B_j)]=0, \\
[Y_i^\b(B_i),Y_j^\b(\bar{B}_j)]=Y_{i+j}^\b(\{B_i,\bar{B}_j\}),\\
[Y_i^\a(A_i),Y_0']=Y_i^\a(\mu^a_-(A_i)),\quad [Y_i^\a(A_i),Y_0'']=0,\\
  [Y_i^\a(A_i),\tilde{Y}_i]=Y_i^\a((2-2i)A_i-\mu^\a_+(A_i)),\\
[Y_i^\b(B_i),Y_0']=0,\quad [Y_i^\b(B_i),Y_0'']=Y_i^\b(\mu^\b_-(B_i)),\\
  [Y_i^\b(B_i),\tilde{Y}_i]=Y_i^\b(\mu^\b_+(B_i)-(2-2i)B_i),
 \end{gather*}
where $\mu^\a_\pm(A)=yA_y\pm tA_t$, $\mu^\b_\pm(B)=zB_z\pm xB_x$.
 \end{theorem}

Both the descending central series and the derived series of $\tilde\g$ stabilize, since
$[\tilde\g,\tilde\g]=[\tilde\g,\g]=[\g,\g]=\g$. The left-hand side of (\ref{PbI}) is an absolute
invariant of the Lie algebra $\tilde\g$, which is a 3-dimensional right extension of $\g$ 
by $\mu^\a_\pm,\mu^\b_\pm$.

\medskip

Higher dimensional versions of the second Pleba\'{n}ski equation are known \cite{PP}.
We can produce some analogs via differential invariants.

For instance, taking 6 copies of $\op{SDiff}(2)$ and applying the above method for
the first Pleba\'{n}ski equation, we obtain, modulo the standard copies of 4-dimensional
equation (\ref{PbI}), the unique 6-dimensional equation on $u=u(x_1,p_1,x_2,p_2,x_3,p_3)$:
 $$
\op{Pf}\begin{bmatrix} 0 & H_{12} & H_{13} \\ -H_{12}^T & 0 & H_{23} \\ -H_{13}^T & -H_{23}^T & 0
\end{bmatrix}=0,
 $$
where $\op{Pf}$ is the Pfaffian,
$H_{ij}=\begin{bmatrix} u_{x_ix_j} & u_{x_ip_j} \\ u_{p_ix_j} & u_{p_ip_j} \end{bmatrix}$
is the $2\times 2$ sub-matrix of $\op{Hess}(u)$ and $H_{ij}^T$ its transpose.

This equation is cubic in 2-jets, and the standard integrability methods are not applicable.
Still it has a huge local symmetry algebra.
The geometry of this equation should be a subject of the further study.


\end{document}